\documentclass[12pt,reqno,toc=bibliography,]{scrartcl}
\usepackage[margin=4cm
]{geometry}
\makeatletter
\DeclareOldFontCommand{\rm}{\normalfont\rmfamily}{\mathrm}
\DeclareOldFontCommand{\sf}{\normalfont\sffamily}{\mathsf}
\DeclareOldFontCommand{\tt}{\normalfont\ttfamily}{\mathtt}
\DeclareOldFontCommand{\bf}{\normalfont\bfseries}{\mathbf}
\DeclareOldFontCommand{\it}{\normalfont\itshape}{\mathit}
\DeclareOldFontCommand{\sl}{\normalfont\slshape}{\@nomath\sl}
\DeclareOldFontCommand{\sc}{\normalfont\scshape}{\@nomath\sc}
\makeatother

\usepackage[utf8]{inputenc}

\usepackage[english]{babel}

\usepackage{physics}
\usepackage{braket}
\usepackage{dsfont}
\usepackage[table,xcdraw]{xcolor}
\usepackage{multirow}
\usepackage{graphicx}
\usepackage{tikz}
\usepackage{amsmath}
\usepackage{amssymb}
\usepackage{amsthm}
\usepackage{mathtools}
\usetikzlibrary{matrix}
\usepackage{appendix}

\usepackage{xcolor}
\usepackage{bbm} 
\newcommand{\1}{\mathbbm{1}} 

\newtheoremstyle{break}
  {5pt}{5pt}%
  {\itshape}{}%
  {\bfseries}{}%
  {\newline}{}%
\theoremstyle{break}

\usepackage[color]{changebar}

\usepackage{tocbasic}

\newtheorem{theorem}{Theorem}
\newtheorem{apptheorem}{Theorem}[section]
\newtheorem{lemma}{Lemma}
\newtheorem{corollar}{Corollary}
\newtheorem{appcor}[apptheorem]{Corollary}

\newtheorem{definition}{Definition}
\newtheorem{prop}{Proposition}

\newcommand{\bbC}{\mathbb{C}}

\newcommand{\bbR}{\mathbb{R}}

\newcommand{\bbN}{\mathbb{N}}
\newcommand{\bbQ}{\mathbb{Q}}
\newcommand{\bbP}{\mathbb{P}}
\newcommand{\bbZ}{\mathbb{Z}}

\newcommand{\calS}{\mathcal{S}}
\newcommand{\calU}{\mathcal{U}}

\newcommand{\calG}{\mathcal{G}}
\newcommand{\calD}{\mathcal{D}}

\newcommand{\calC}{\mathcal{C}}

\newcommand{\note}[1]{}

\setkomafont{author}{\sffamily}
\setkomafont{title}{\sffamily}
\makeatletter
\patchcmd{\@maketitle}{\huge}{\LARGE}{}{}
\makeatother

\RedeclareSectionCommand[tocbeforeskip=1ex plus 1pt minus 1pt]{section}
\makeatletter
\renewcommand*{\@pnumwidth}{1em}
\makeatother

\deftocheading{toc}{} 

\title{\vspace*{-22pt}Hay from the haystack: explicit examples of exponential quantum circuit complexity}

\usepackage{authblk}

\author{Yifan Jia,\ Michael M. Wolf}

\affil{\small Department of Mathematics, Technische Universit\"at M\"unchen, Germany\\
Munich Center for Quantum Science and Technology (MCQST), Germany\\
{\footnotesize\texttt {\{yifan.jia, m.wolf\}@tum.de}}}

\date{\vspace*{-30pt}\small \vspace*{-12pt} }

\begin{document}  

\maketitle

\tableofcontents

\begin{abstract}\vspace*{5pt} \noindent{\bf Abstract: }\small
The vast majority of quantum states and unitaries have circuit complexity exponential in the number of qubits. In a similar vein, most of them also have exponential minimum description length, which makes it difficult to pinpoint examples of exponential complexity. In this work, we construct examples of constant  description length but exponential circuit complexity. We provide infinite families such that each element requires an exponential number of two-qubit gates to be generated exactly from a product and where the same is true for the approximate generation of the vast majority of elements in the family.  The results are based on sets of large transcendence degree and discussed for tensor networks, diagonal unitaries and maximally coherent states.\end{abstract}

\section{Introduction}

In 1949 a simple counting argument led Claude Shannon to the observation that almost every Boolean function $f:\{0,1\}^n\rightarrow\{0,1\}$ requires an exponential size Boolean circuit \cite{shannon1949synthesis}: the number of functions simply outgrows the number of circuits of sub-exponential size. As simple as this argument is, it appears to be hopelessly difficult to find explicit examples for which a better-than-linear lower bound can be proven. Even a $4n$ linear lower bound is currently out of reach \cite{blum1983boolean,B3n}. Apart from more sophisticated (and more severe) reasons there is an obvious obstacle to pinpointing `exponential hay' in this haystack:  for most functions not only the circuit complexity but also the descriptive complexity is exponential. That is, we are not able to talk about them efficiently. 

The quantum world exhibits a similar problem. It is textbook knowledge  that most states and unitaries on an $n$-partite Hilbert space have exponential quantum circuit complexity (cf. p.198 in \cite{nielsen2002quantum} and \cite{knill1995approximation}). In fact, the very possibility of this exponential size played a crucial role in the birth of a  vastly growing research direction \cite{Susskind}. Yet, as in the classical case, while linear lower bounds are known (e.g., for states with topological order \cite{bravyihastingsverstraete,nospeedupdiss}), super-linear lower bounds remain elusive.

The aim of this paper is to provide explicit examples of exponential quantum circuit complexity. In order to achieve this, we follow a simple idea: small circuits involve a small number of parameters and can thus only lead to matrices with a small number of `independent entries'---made precise by the \emph{transcendence degree}. Large transcendence degree therefore implies large circuit complexity. The transcendence degree is a tool that is well-known and -established  in algebraic complexity theory and was introduced to the analysis of arithmetic circuits and formulas in \cite{Trdegorigin2} and \cite{Trdegorigin1}.

The price one has to pay when resorting to such an algebraic tool, is that it is difficult to obtain results that are reasonably stable under perturbations. For this reason, we will first focus on \emph{exact} generation of states and unitaries. In the second part of the paper,  we will then extend the focus and discuss results on approximate generation, but those will only achieve our goal with some caveats. As pointed out in \cite{aaronson2016complexity}, stronger results would immediately imply complexity theoretic separations such as $\mathsf{BQP}\neq\mathsf{P}^{\mathsf{\#P}}$, at least when obtained for states.\medskip

After introducing some basic notions and concepts, we will first study the transcendence degree and its relation to circuit complexity with some excursions to tensor networks. Then we will construct examples of large transcendence degree and thus large circuit complexity and finally look into some hardness-of-approximation implications.

\section{Preliminaries}\label{preliminaries}
  \subsection{Circuit complexity}\label{sec:prelim1}
   In this section, we introduce some fundamental notions and definitions concerning circuit complexity. The \emph{quantum circuit complexity} of a state or unitary  on an $n$-partite Hilbert space quantifies the number of local building blocks that are required for their implementation. The Hilbert space that we will consider is $(\bbC^d)^{\otimes n}\simeq\bbC^{d^n}$ for arbitrary $d\in\bbN$. That is, the local system will be a qu$d$it and the `local building blocks' will be one or two-qudit gates. More specifically, we fix a \emph{gate set} $\calG\subset U(d^2)$ that is assumed to be \emph{ universal} in the sense that the group it generates is dense in $U(d^2)$.\footnote{From a physical perspective it suffices to generate a dense subset of the projective unitary group $U(d^n)/U(1)$. However, we will mostly disregard this quotient since it would make many of the statements more cumbersome without changing their essence. In particular, it would not change the $n$-dependence of the results since a global phase can always be changed locally.} This means that for any $U\in U(d^n)$ and any $\epsilon>0$ there is a finite sequence of unitaries $U_1,\ldots,U_g\in U(d^n)$ such that 
   \begin{equation}\label{eq:Uapprox}
       \|U-U_1U_2\cdots U_g\|_\infty\leq \epsilon,
   \end{equation}
   where each $U_i$ is up to a permutation of tensor factors of the form $U_i\simeq V_i\otimes\1_d^{\otimes(n-2)}$ with $V_i\in\calG$. The smallest number $g$ for which this is possible then defines the quantum circuit complexity $\calC_\epsilon(U,\calG)$. We will often consider the entire unitary group $U(d^2)$ as a gate set and write $\calC_\epsilon(U):=\calC_\epsilon(U,U(d^2))$. Clearly, $\calC_\epsilon(U)\leq \calC_\epsilon(U,\calG)$ holds for any $\calG\subset U(d^2)$. By the Solovay-Kitaev theorem \cite{Kitaevbook}, an inequality in the opposite direction can be proven as well.
   That is, if an $n$-qudit unitary is hard to approximate  with any finite universal gate set, then the same is true with the gate library consisting of all one- and two-qudit gates.
   This well-known fact is spelled out in the following Lemma for the case of \emph{efficiently universal} gate sets, i.e., those which enable an $\epsilon$-approximation within $SU(d^2)$ with $g=O(\ln 1/\epsilon)$.  Efficiently universal gate sets  with cardinality $|\calG|\leq 3(d^2-1)$ were constructed in \cite{harrow2002efficient}.

   \begin{lemma}[Change of gate sets]\label{lem:SK}
    For every efficiently universal  gate set $\calG\subset SU(d^2)$ there is a constant\footnote{A more explicit expression for $c$ can be found  in Eq.(14) of \cite{harrow2002efficient}.} $c$ such that for every $U\in SU(d^n)$ and $\epsilon>0$:
    \begin{equation}\label{eq:SKbound}
        \calC_{2\epsilon}(U,\calG)\leq c\;\calC_\epsilon(U)\; \ln\frac{\calC_\epsilon(U)}{\epsilon}.
    \end{equation}
   \end{lemma}
   \emph{Remarks:} (i) It is unclear how restrictive the assumption is for $\calG$ to be efficiently universal, as opposed to merely universal within $SU(d^2)$. In the latter case, a bound similar to Eq.(\ref{eq:SKbound}) holds if the logarithm is taken to some power $p$ \cite{varju2012random}. Depending on assumptions (e.g. whether $\calG$ includes inverses) proven bounds for $p$ range from $p= 3$ \cite{inversefree1} over $p=3.97$ \cite{dawson2005solovay} to $p=O(\ln d)$ \cite{bouland2021efficient}. (ii) The restriction to the special (or projective) unitary group is necessary for the Solovay-Kitaev theorem. However, since a global phase can be changed locally, without additional $n$-dependence, the distinction between $SU(d^n)$ and $U(d^n)$ is not relevant for our purpose.
   \begin{proof}
    $\calC_\epsilon(U)=g$ implies  Eq.(\ref{eq:Uapprox}) with each of the $U_i$ acting non-trivially only on at most two qudits. Approximating each of the $U_i$'s in turn via  $\|U_i-U_{i,1}\cdots U_{i,m}\|_\infty\leq\nu$ (for some $\nu$ to be chosen later) by  $m=c\ln(1/\nu)$ unitaries from $\calG$,
    we obtain  
    \begin{equation}\label{eq:Uprodapprox}
        \Big\|U-\prod_{i=1}^g\prod_{\alpha=1}^m U_{i,\alpha}\Big\|_\infty \leq \epsilon+g \nu
    \end{equation}
    when reading the product in the right order. Eq.(\ref{eq:Uprodapprox}) then proves that $\calC_{\epsilon+g\nu}(U,\calG)\leq mg$, which implies the claimed inequality when inserting $m$ and choosing $\nu=\epsilon/g$.
   \end{proof}

    If the entire unitary group $U(d^2)$ is chosen as gate library, then even in the exact case ($\epsilon=0$)  every unitary $U\in U(d^n)$  has finite circuit complexity  $\calC_0(U)=O (d^{2n})$ \cite{optimalquditcircuits}.
    Moreover, since the set $\{V\in U(d^n)|\;\calC_0(V)< g\}$ is compact, we have that 
    \begin{equation}
        \calC_0(U)=g \quad\Rightarrow\quad\exists\epsilon>0:\;\calC_\epsilon(U)=g.
    \end{equation}
    Here we can choose any positive $\epsilon<\inf\{\|U-V\|_\infty|\; V\in U(d^n)\wedge \calC_0(V)< g\}$ since compactness of the variational set together with the continuity of the norm guarantee that the infimum is non-zero.
    
    The definition of the mentioned quantum circuit complexities is easily generalized from unitaries to states: the circuit complexity of a state is  defined as the minimum exact circuit complexity of all unitaries that generate the considered state from a product state (up to $\epsilon$). In particular, for a pure state given by a unit vector $|\varphi\rangle\in(\bbC^d)^{\otimes n}$:
    \begin{equation}
        \calC_0(|\varphi\rangle ):=\min\left\{\calC_0(U)\;|\;\exists U\in U(d^n):|\varphi\rangle=U(|0\rangle)^{\otimes n}\right\}.
    \end{equation}
    Similarly, for a density operator $\rho $ acting on $ (\bbC^{d})^{\otimes n}$ we define
      \begin{equation}
        \calC_0(\rho ):=\min\left\{\calC_0(|\varphi\rangle)\;\big|\;\exists m\in\mathbbm{N}_0\;\exists|\varphi\rangle\in(\bbC^{d})^{\otimes (n+m)}:  \rho={\rm tr}_{\bbC^{d^m}}[|\varphi\rangle\langle\varphi|]\right\}.
        \end{equation}

Other notions of circuit complexity could be defined by considering ancillas also for the cases of pure states or unitaries and/or by allowing for measurements. However, these can all be bounded from below by the minimal size of the respective tensor network description.


    \begin{figure}[t]
    \makebox[\textwidth]{
    \includegraphics[width=14cm]{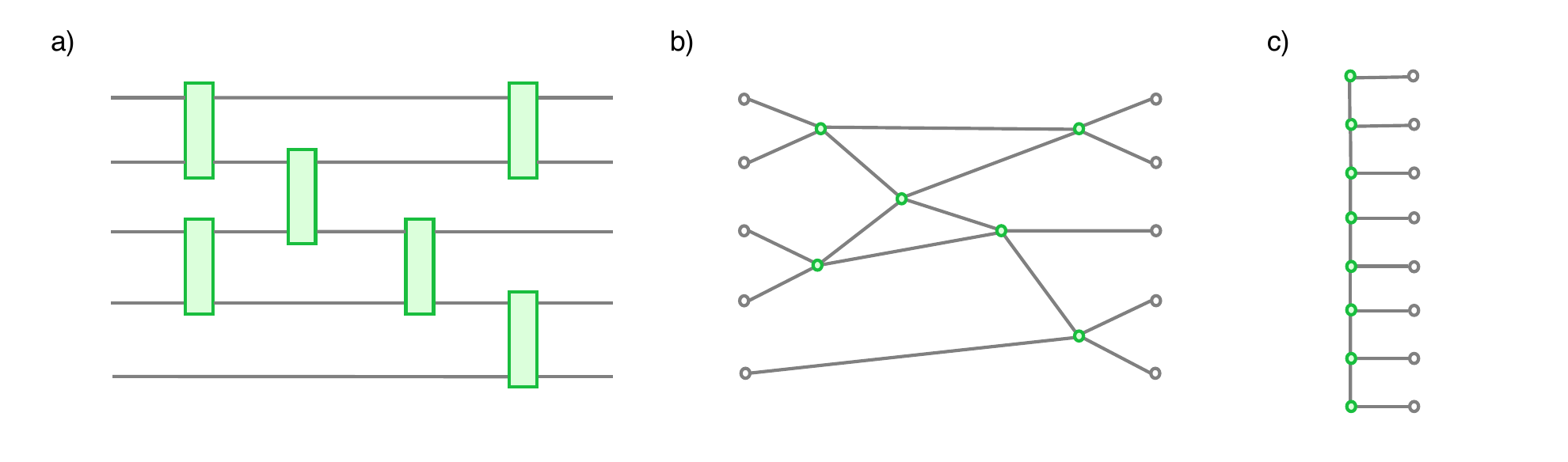} }
    \caption{a) Graphical depiction of a circuit that involves six gates and represents a unitary on $(\bbC^d)^{\otimes 5}$. b) The graph of the corresponding tensor network representation has six internal vertices of degree four. These are assigned tensors of the same degree that are contracted according to the edges. c) The graph corresponding to a \emph{matrix product state} on $(\bbC^d)^{\otimes 8}$. } \label{fig:tensors}
\end{figure}

  \subsection{Tensor network representations}\label{sec:tensors} 
We will employ a perspective on tensors that is targeted on their explicit numerical representation. That is, we will  consider tensors represented in a given product basis so that they become multi-dimensional arrays of complex numbers. For simplicity, we assume that each dimension has the same size, say $d$, so that a tensor of order $n$ becomes a map $\Psi:\bbZ_d^n\rightarrow\bbC$, i.e., for any $j=(j_1,\ldots,j_n)\in\bbZ_d^n$, which specifies a location in the array, $\Psi(j)$ is the corresponding complex number. A unitary $U$ that acts on $(\bbC^d)^{\otimes n}$, when represented in a product basis, then becomes a tensor of order $2n$ that maps $\bbZ_d^n\times\bbZ_d^n\ni (j,k)\mapsto\langle j_1,\ldots,j_n|U|k_1,\ldots,k_n\rangle$.
   
A \emph{tensor network representation} of a tensor  $\Psi:\bbZ_d^n\rightarrow\bbC$ has two ingredients: a graph and an assignment of tensors to the `internal' vertices of that graph. More precisely, the set $V$ of vertices of the graph is assumed to be a  union $V=V_i\cup V_e$ of $|V_e|=n$ `external' (or `physical') vertices that have degree one and `internal' (or `virtual') vertices in $V_i$ of higher degree. To each edge $e$ there is assigned a dimension parameter $d_e\in\bbN$, which is constrained to $d_e=d$ if the edge connects to an external vertex. In addition, each internal vertex $v\in V_i$  gets assigned a tensor $\psi_v:\bbZ_{d_{e_1}}\times\cdots\times\bbZ_{d_{e_\delta}}\rightarrow\bbC$, where the order $\delta$ equals  the degree of the vertex and the $e_i$'s are the vertex' edges. Finally, $\Psi$ is represented by the tensor network if it is obtained from the product $\prod_{v\in V_i}\psi_v$ after contracting (i.e., summing) over all indices that correspond to edges that connect internal vertices. 
   
   If $N=|V_i|$ is the number of internal vertices, $\delta$ is a uniform upper bound on their degree and $D:=\max_e\{d_e\}$ defines the `maximal bond dimension', then the number of complex numbers that specify the tensor network is at most
   \begin{equation}\label{eq:tensornetworkparameters}
       ND^\delta.
   \end{equation}
The representation of a unitary $U$ on $(\bbC^d)^{\otimes n}$ in terms of a circuit of size $N$ can be regarded as a special case of a tensor network representation with $D=d$ and $N$ internal vertices, each of degree $\delta=4$ (see Fig.\ref{fig:tensors}).
Similarly, a state vector $|\varphi\rangle\in(\bbC^d)^{\otimes n}$  generated from a product $|0\rangle^{\otimes(n+m)}$ via a unitary circuit of size $N$ and conditioned on outcomes of local measurements  on an $m$-partite ancilla has a tensor network representation with $|V_i|=N$, $D=d$ and $\delta\leq 4$. 
   
\section{Transcendence degree as lower bound for complexity}\label{sec:transcendencedegree}
For any set $S$ of complex numbers, we denote by $\bbQ(S)$ the number field generated by $S$ over $\bbQ$, i.e., the minimal field extension of the rational numbers that contains $S$. The algebraic closure of a field will be indicated by a bar, so in particular $\overline{\bbQ}$ means the field of complex algebraic numbers.

A set of numbers $\{a_1,\ldots,a_n\}\subset \bbC$  is called \emph{algebraically independent} if  the only polynomial $p\in\overline{\bbQ}[x_1,\ldots,x_n]$ with algebraic coefficients that satisfies $p(a_1,\ldots,a_n)=0$ is the zero-polynomial. This leads to the central quantity of our analysis:
\begin{definition}[Transcendence degree]\label{def:transcendence_degree}
For $S\subseteq\bbC$ we define the \emph{transcendence degree} $\gamma(S)\in\bbN\cup\{0,\infty\}$ to be the maximal number of algebraically independent elements of $S$.
\end{definition}
\emph{Remark:} We will often apply $\gamma$ to structured sets such as matrices or tensors. In that case, the additional structure is simply ignored so that only the set of entries is considered. More specifically, for $A\in\bbC^{n\times n}$ we have $\gamma(A):=\gamma\big(\{A_{i,j}\}_{i,j=1}^n\big)$ and for a map $\Psi:X\rightarrow\bbC$ we define $\gamma(\Psi):=\gamma\big(\{\Psi(x)\}_{x\in X}\big)$.\medskip

The \emph{transcendence degree} is well-known and studied in the theory of field extensions (see for instance Chap.V in \cite{morandi2012field}). By definition, the maximal number of algebraically independent elements does not change when going from $S$ to the field it generates over $\bbQ$ or its algebraic closure. Hence,
\begin{equation}\label{eq:gammafield}
    \gamma(S)=\gamma\left(\bbQ(S)\right)=\gamma\left(\overline{\bbQ(S)}\right).
\end{equation}
The following  summarizes  properties of $\gamma$ that turn out to be useful in our context:

\begin{lemma}[Properties of the transcendence degree]\label{lem:gammaprops}\mbox{}
\begin{enumerate}\vspace*{-10pt}
    \item The matrix product $AB$ of two matrices $A$ and $B$ satisfies
    \begin{equation}
        \gamma(AB)\leq\gamma(A)+\gamma(B).\label{eq:additivity}
    \end{equation}
    \item Every unitary $U\in\bbC^{d\times d}$ satisfies $\gamma(U)=\gamma(U\cup U^*)\leq d^2$. Moreover, the set of unitaries for which $\gamma(U)\neq d^2$  is a null set w.r.t. the Haar measure on $U(d)$.
    \item If $A\in\bbC^{d\times d}$ has spectrum ${\rm{spec}(A)}$, then $\gamma(A)\geq \gamma\big({\rm{spec}(A)}\big)$.
\end{enumerate}
\end{lemma}
\begin{proof}
    (1) Using Eq.(\ref{eq:gammafield}) together with the definition of $\gamma$, we obtain $\gamma(AB)=\gamma(\bbQ(AB))\leq\gamma(\bbQ(A\cup B))=\gamma(A\cup B)\leq\gamma(A)+\gamma(B)$.
    
    (2) Clearly, $\gamma(U)\leq d^2$ as this is maximal number of distinct entries.  Due to Eq.(\ref{eq:gammafield}) we can assume that $-1\not\in{\rm spec}(U)$ since we can always multiply $U$ with an algebraic number $\omega$ of modulus one and use that $\gamma(U)=\gamma(\omega U)$ as well as $\gamma(U\cup U^*)=\gamma(\omega U\cup\overline{\omega}U^*)$. Under this assumption, the Cayley transform $C(X):=(\1+X)(\1-X)^{-1} $, which happens to be an involution, provides a skew-hermitian matrix $A:=C(U)$ such that $U=C(A)$. Using Gauss elimination to compute the inverse and again Eq.(\ref{eq:gammafield}), we see that the Cayley transform does not change the transcendence degree. So $\gamma(U)=\gamma(A)$. 
    Moreover, $U^*=C(-A)$ implies that each matrix element of $U^*$ is contained in $\bbQ(A)$ as well. Hence, $\gamma(U\cup U^*)=\gamma(A)=\gamma(U)$.
    
    In order to see that the bound $\gamma(U)\leq d^2$ generically holds with equality it suffices to show this for $\gamma(A)$ since the Cayley transform maps Lebesgue null sets into Haar null sets and vice versa. $A$, however, can be parametrized by $d^2$ real, unconstrained parameters. If these are not all algebraically independent, then $d^2-1$ of them determine the last one up to a countable number of alternatives. Due to $\sigma$-additivity of the Lebesgue measure, $\gamma(A)\neq d^2$ is thus a null set. Hence, $\gamma(U)\neq d^2$ also characterizes a null set w.r.t. the Haar measure. 
    
    (3) $\lambda\in{\rm spec}(A)$ means that $p(\lambda)=0$ holds for the characteristic polynomial $p(\lambda):=\det(\lambda\1-A)$. The fact that $p$'s coefficients are in $\bbQ(A)$ then implies that ${\rm spec}(A)\subset\overline{\bbQ(A)}$ and therefore, by Eq.(\ref{eq:gammafield}), $ \gamma\big({\rm{spec}(A)}\big)\leq\gamma\big(\overline{\bbQ(A)}\big)=\gamma(A)$.
\end{proof}
With this Lemma at hand we can now prove lower bounds on the circuit complexity in terms of the transcendence degree:
\begin{theorem}[Transcendence degree vs. circuit complexity]\label{thm:complexitybounds}
    Let $U$ be a unitary matrix, $|\varphi\rangle$ a state vector and $\rho$ a density matrix, all three  represented in a product basis of $(\bbC^d)^{\otimes n}$. Then: 
    \begin{eqnarray}
    \gamma\big({\rm spec}(U)\big)\ \leq\ \gamma(U) &\leq& d^4\;\calC_0(U),\label{eq:circuittransc1}\\
    \gamma\big({\rm spec}(\rho)\big)\ \leq\ \gamma(\rho) &\leq& d^4\;\calC_0(\rho),\label{eq:circuittranscrho}\\
    \gamma\big(|\varphi\rangle\big) &\leq& d^4\;\calC_0\big(|\varphi\rangle\big).\label{eq:circuittransc2}
    \end{eqnarray}
\end{theorem}
\begin{proof}
    The leftmost inequalities in Eqs.(\ref{eq:circuittransc1},\ref{eq:circuittranscrho}) are from {\it{3.}} in Lemma \ref{lem:gammaprops}. For the second inequality in Eq.(\ref{eq:circuittransc1}) suppose $U=\prod_{k=1}^{\calC_0(U)} U_k$ is a minimal circuit for $U$ where each $U_k$ is (up to permutation of tensor factors) of the form $U_k=\1\otimes V_k$ with $V_k\in U(d^2)$. Using that $\gamma(\1\otimes V_k)=\gamma(V_k)$ and applying {\it{1.}} and {\it{2.}} from Lemma \ref{lem:gammaprops} we obtain
    $$\gamma(U)\leq \sum_{k=1}^{\calC_0(U)} \gamma(U_k) \leq d^4\;{\calC_0(U)}. $$
    This in turn implies Eq.(\ref{eq:circuittransc2}) for $|\varphi\rangle=U|0\rangle^{\otimes n}$ via $\gamma\big(|\varphi\rangle\big)\leq \gamma(U)+\gamma\big(|0\rangle^{\otimes n}\big)$ as $\gamma\big(|0\rangle^{\otimes n}\big)=0$. 
    
    For the second inequality in Eq.(\ref{eq:circuittranscrho}), suppose $\rho$ is the partial trace of $|\varphi\rangle=U|0\rangle^{\otimes{(n+m)}}$ with $\calC_0(\rho)=\calC_0(|\varphi\rangle)=\calC_0(U)$. Then,  using that $\gamma(U\cup U^*)=\gamma(U)$ as shown in {\it{2.}} in Lemma \ref{lem:gammaprops} and exploiting the right inequality of Eq.(\ref{eq:circuittransc1}), we obtain:
    \begin{equation*}
        \gamma(\rho)\leq\gamma(|\varphi\rangle\langle\varphi|)\leq\gamma(\varphi \cup\bar{\varphi})\leq \gamma(U\cup U^*)=\gamma(U)\leq d^4\;\calC_0(\rho).
    \end{equation*}
\end{proof}
A similar result holds if we employ the more general perspective of tensors and tensor network representations discussed in Sec.\ref{sec:tensors}. To put it simply, the number of parameters of the network is an upper bound on the transcendence degree of the tensor it represents:
\begin{theorem}[Transcendence degree of tensor networks]\label{thm:tensornetworkbound}
Let $\Psi:\bbZ_d^n\rightarrow\bbC$ be a tensor that admits a tensor network representation with maximal bond dimension $D$ using $N$ internal vertices of degree at most $\delta$. Then
\begin{equation}
    \gamma(\Psi)\leq  N D^\delta.
\end{equation}
\end{theorem}
\begin{proof}
As discussed in Sec.\ref{sec:tensors} and Eq.(\ref{eq:tensornetworkparameters}) the set $S:=\{\psi_v(j)\}_{v,j}$ of complex numbers that specify the considered tensor network contains at most $ND^\delta$ elements. As $\Psi$ is built up by algebraic operations from $S$, we have that $\gamma(\Psi)\leq\gamma(\bbQ(S))=\gamma(S)\leq ND^\delta$. 
\end{proof}

Before moving to explicit examples, we will have a brief look at the transcendence degree of density matrices that are represented as \emph{Gibbs states} of algebraic Hamiltonians. In this case, the transcendence degree of the density matrix is essentially equal to the number of $\bbQ$-linearly independent eigenvalues  of the  Hamiltonian:
\begin{prop}[Transcendence degree of Gibbs states]
 If $H\in\overline{\bbQ}^{d\times d}$ is Hermitian, and its spectrum spans a $\bbQ$-vector space of dimension $\dim_\bbQ\big({\rm spec}(H)\big)$, then $\rho:=e^H/\tr[e^H]$ satisfies
    \begin{equation}
        \gamma\left(\rho\right)\leq \dim_\bbQ\big({\rm spec}(H)\big)\leq \gamma\left(\rho\right)+1.
    \end{equation}
\end{prop}
\begin{proof}
    With $E:=\{e^\lambda|\lambda\in{\rm spec}(H)\}$ and $c:=\tr[e^H]^{-1}\in\bbQ(E)$ we  use the spectral decomposition  $\rho=c\;UDU^*$, where the diagonal matrix $D$ has entries in $E$ and the unitaries $U,U^*$ have entries in $\overline{\bbQ(H)}=\overline{\bbQ}$. So $\gamma(\rho)=\gamma(cD)$. Moreover, Cor.\ref{cor:E} implies that $\gamma(D)=\dim_\bbQ\big({\rm spec}(H)\big)$. From here the claim follows since
    \begin{eqnarray}
        \gamma(cD)&\leq&\gamma\big(\bbQ(E)\big)=\gamma(D),\label{eq:Gibbs1}\\
        \gamma(D)&\leq&\gamma(c^{-1})+\gamma(cD)=1+\gamma(cD)\label{eq:Gibbs2}.
    \end{eqnarray}
    Here, Eq.(\ref{eq:Gibbs1}) has used that $c\in\bbQ(E)$ and Eq.(\ref{eq:Gibbs2}) is based on writing $D=c^{-1}cD$ and using the sub-additivity of the transcendence degree specified in Eq.(\ref{eq:additivity}).
\end{proof}
Combining this with Theorem \ref{thm:complexitybounds} we get that the dimension of the spectrum of an algebraic Hamiltonian provides a lower bound on the exact circuit complexity of its Gibbs state.

\section{Explicit examples and their complexities}\label{sec:examples}
Since we have bounded the circuit complexity below by the transcendence degree, obtaining explicit examples of large circuit complexity is reduced to finding explicit sets of large transcendence degree. According to \emph{Schanuel's conjecture} (Chap.21 in \cite{MurtyRath}), every set of complex numbers $\alpha_1,\ldots,\alpha_n$ that are linearly independent over $\bbQ$ should lead to 
\begin{equation}
    \gamma\big(\{\alpha_1,\ldots,\alpha_n,e^{\alpha_1},\ldots,e^{\alpha_n}\}\big)\geq n.
\end{equation}
In fact, most of the known transcendence results would follow from here. We will, however, not resort to this conjecture. The explicit examples provided in the following theorem rather follow from   known results that can be viewed as proven special cases or implications of Schanuel's conjecture.
\begin{theorem}[Explicit sets of large transcendence degree]\label{thm:largesets} \mbox{}
Consider any $n,d\in\bbN$ with $d\geq 2$.
\begin{enumerate}
    \item  Let $p_1, \dots,p_n$  be  distinct prime numbers and $t\in\overline{\bbQ}\setminus\{0\}$. Then
	\begin{equation}
	    \gamma\left(\big\{e^{it\phi(j)}\big\}_{j\in\bbZ_d^n}\right)=d^n\quad \text{for}\quad 	\phi(j):=\left(\prod\limits_{k=1}^{n}p_k^{j_k}\right)^{\frac{1}{d}}.\label{eq:mainexampleset}
	\end{equation}
    \item If $\alpha\neq 0,1$ is algebraic and $m>1$ is a square-free integer, then:
    \begin{equation}
        \gamma\left(\Big\{\alpha^{m^{\frac{k}{d}}}\Big\}_{k=1}^{d-1}\right)\geq\frac{d}2.\label{eq:mainexampleset2}
    \end{equation}
\end{enumerate}
\end{theorem}
\begin{proof}
    (1) As shown in Cor.\ref{cor:Besicovitch} in Appendix \ref{app:tnt}, Besicovitch's theorem implies that $\{\phi(j)\}_j$ is a set of $d^n$ algebraic numbers that are linearly independent over $\bbQ$. Clearly, this remains true when multiplied by $it$. By the Lindemann-Weierstrass theorem (Thm.\ref{thm:LW} in Appendix \ref{app:tnt}), the exponential function then lifts $\bbQ$-linear independence to algebraic independence.
    
    (2) is implied by the Diaz-Philippon theorem (Thm.\ref{thm:DiazPhilippon} in Appendix \ref{app:tnt}) when using that $m^{1/d}$ is an algebraic number of degree $d$. The latter can for instance be seen as a consequence of Besicovitch's theorem: if $m=p_1\cdots p_s$ is the prime factorization of $m$ and we assume that there would be a polynomial $p\in\bbQ[x]$ of degree less than $d$ that has $m^{1/d}$ as a root, then $P(p_1^{1/d},\ldots,p_s^{1/d}):=p\big((p_1\cdots p_s)^{1/d}\big)=0$ would contradict  Besicovitch's theorem.
\end{proof}

Combining the sets of large transcendence degree of Thm.\ref{thm:largesets} with the complexity bounds of Thm.\ref{thm:complexitybounds} and Thm.\ref{thm:tensornetworkbound} then finally leads to explicit examples of exponential complexity. The following corollary makes this explicit for the first construction of Thm.\ref{thm:largesets}:
\begin{corollar}[Explicit examples of exponential circuit complexity]\label{cor:exexcc}
Let $\{|j\rangle\}_{j\in\bbZ_d^n}$ be an orthonormal product basis of $(\bbC^d)^{\otimes n}$, $d\geq 2$,  $t\in(\overline{\bbQ}\cap\bbR)\setminus\{0\}$ and $\phi(j)$ as in Eq.(\ref{eq:mainexampleset}). 
\begin{enumerate}
    \item For $H:=\sum_{j} \phi(j)|j\rangle\langle j|$ the unitary $U:=e^{itH}$ satisfies
    \begin{equation}
        \calC_0(U) \geq d^{n-4}.\label{eq:C0Ulower}
    \end{equation}
    \item The state vector $|\varphi\rangle:=d^{-n/2}\sum_j \exp[it\phi(j)]\; |j\rangle$ satisfies
    \begin{equation}
        \calC_0(|\varphi\rangle) \geq d^{n-4}.
    \end{equation}
    \end{enumerate}
    Moreover, both $U$ and $|\varphi\rangle$ do not admit a tensor network representation (as specified in Sec.\ref{sec:tensors}) with less than $d^n$ parameters.
\end{corollar}
\emph{Remark:} Using the spectral inequality of Thm.\ref{thm:complexitybounds} we get that Eq.(\ref{eq:C0Ulower}) does not only hold for $U$ but for every unitary that has the same spectrum as $U$.
In other words, the result is stable with respect to perturbations of the eigenbasis. However, this stability as well as the related simplicity of the examples comes at a price: the complexity is, although exponential in $n$, not maximal, as the maximum complexity grows as $d^{2n}$ rather than as $d^n$. The following  shows that examples of maximal complexity $\Omega(d^{2n})$ can be constructed using the same toolbox.\vspace*{5pt}

\emph{Modified example:} For $i,j\in \bbZ_d^n$ we define $$\phi(i,j):=\left(\prod\limits_{k=1}^{n}p_k^{i_kd+j_k}\right)^{\frac{1}{d^2}},$$
 where $p_1,\ldots, p_n$ are again distinct primes. According to Thm.\ref{thm:largesets}, the matrix $A$ with coefficients $A_{ij}:=e^{\phi(i,j)}$ then has transcendence degree $d^{2n}$. The Hermitian matrix  $H_A:=(A+A^*)+i(A-A^*)$ then satisfies $\gamma(H_A)=\gamma(A)$ and via the Cayley transform of $H_A$ we obtain a unitary matrix $U_A\in U(d^{2n})$ defined by $U_A:=(H_A-i\1)(H_A+i\1)^{-1}$. As the Cayley transformation does not change the transcendence degree either (see the proof of Lemma \ref{lem:gammaprops}), $\gamma (U_A)=d^{2n}$, and it follows immediately from Theorem \ref{thm:complexitybounds} that $\calC_0(U_A)\geq d^{2n-4}$. Hence, these examples have maximal complexity up to a constant factor. However, since their construction is more cumbersome and the coefficients are not as directly given as in the examples of Cor. \ref{cor:exexcc}, we will work with the latter in the following. \vspace*{5pt}

Clearly, one can construct other examples by either using Eq.(\ref{eq:mainexampleset}) in a different way or by resorting to Eq.(\ref{eq:mainexampleset2}). Instead of doing this, however, we broaden the focus and look into approximability of the  examples provided in Cor.\ref{cor:exexcc}.

\section{Hardness of approximation}
So far, we have considered the exact (i.e., $\epsilon=0$) circuit complexity, albeit with respect to a continuous gate set. As mentioned in Sec.\ref{sec:prelim1}, for every unitary $U\in U(d^n)$ there is an $\epsilon>0$ so that $\calC_\epsilon(U)=\calC_0(U)$. Replacing this purely topological statement by an explicit quantitative bound, however, appears to be more than difficult. In order to be able to make any quantitative statement for \emph{a given} $\epsilon>0$ we will consider sets of examples of large $\calC_0$-complexity rather than individual instances. In a nutshell, we will first quantify the statement that `most diagonal unitaries and most maximally coherent states are hard to approximate' and then show, by a uniform density argument, that the same is true for each set of  examples. 

\subsection{Diagonal unitaries}\label{sec:approximatediagonal} In this section we employ a standard counting argument \cite{knill1995approximation,nielsen2002quantum} to quantify hardness of approximation for generic diagonal unitaries. In order to get a countable set, we consider a finite gate library $\calG\subset U(d^2)$. By Lemma \ref{lem:SK}, however, the derived results can easily be transferred to the case of any other gate set, including infinite ones.

The following theorem makes the statement `almost all diagonal unitaries are hard to approximate' precise, by showing that the set that can be approximated by circuits of sub-exponential size $O(d^n/\ln n)$ is double-exponentially small. 

\begin{theorem}[Hardness of approximating  diagonal unitaries]\label{thm:harddiagonals}
    For $\epsilon\in[0,1)$ and a  universal gate set $\calG\subset U(d^2)$ with $k\in\bbN$ elements, let  $\mathcal{U}_g$ be the set of unitaries on $(\bbC^{d})^{\otimes n}$ whose complexity is bounded by $g$ so that 
    \begin{equation}
        \mathcal{U}_g:=\big\{U\;|\;\calC_\epsilon(U,\calG)\leq g\big\}\quad\text{with }\quad g:=\left(\frac{\ln\frac{\pi}{2}}{2 k}\right)\frac{d^n}{\ln n}\ .
    \end{equation} 
    \begin{enumerate}
        \item If $\calD$ is the diagonal subgroup of $U(d^n)$ and $\mu$ its normalized Haar measure, then
        \begin{equation}\label{eq:arcsinbound}
            \mu\big(\calD\cap\calU_g\big)\leq\big(\arcsin(\epsilon/2)\big)^{d^n}.
        \end{equation}
        \item Let $\calD_\pm$ be the  subgroup of $U(d^n)$ that contains all diagonal unitaries with diagonal entries $\pm 1$. The fraction of $\calD_\pm$ that is contained in $\calU_g$ is at most 
        \begin{equation}
            \left(\frac{\pi}{4}\right)^{d^n}.
        \end{equation}
    \end{enumerate}
\end{theorem}
\emph{Remarks:} Note that $\arcsin(\epsilon/2)\simeq\epsilon/2$ up to $O(\epsilon^3)$. A result closely related to (1) can be found in \cite{Nielsengeodesic} formulated in terms of geodesic lengths.
\begin{proof} (1) 
Write $U\in \calD$ as $U=\text{diag}\{\exp(i\alpha_1),\dots,\exp(i\alpha_{d^n})\}$ so that we can identify the Haar measure on $\calD$ with the normalized Lebesgue measure on the cube $[0,2\pi)^{d^n}$, which contains the vector of phases $\alpha$.  Consider two diagonal unitary matrices $U$ and $U'$ with distance $\epsilon$ with respect to the $\infty$-norm. Their diagonal elements satisfy $|e^{i\alpha'_j}-e^{i\alpha_j}|\leq \epsilon$, which by a simple geometric argument implies 
	\begin{align}
	\inf\limits_{k\in\bbZ}|\alpha'_j-\alpha_j-2\pi k|\leq \tilde{\epsilon}:= 2\arcsin(\frac{\epsilon}{2}).
	\end{align}
So if $B_\epsilon$ is an $\epsilon$-ball around $U$, then the normalized Haar measure of its intersection with $\calD$ can be bounded by
\begin{equation}
    \mu\big(\calD\cap B_\epsilon\big)\leq \frac{1}{(2\pi)^{d^n}}\int_{[-\tilde{\epsilon},\tilde{\epsilon}]^{d^n}}d\alpha_1\ldots d\alpha_{d^n}=\left(\frac{\tilde{\epsilon}}{\pi}\right)^{d^n}.
\end{equation}

The number of unitaries that can exactly be represented by a circuit of $g$ two-qudit gates that are taken from $\calG$ is at most $\binom{n}{2}^{kg}$. If we consider an $\epsilon$-ball  around each of them we obtain the set $\calU_g$. Hence,

\begin{eqnarray}
    \mu\big(\calD\cap\calU_g\big) &\leq& \binom{n}{2}^{kg} \mu\big(\calD\cap B_\epsilon\big)\label{eq:countingbound1}\\
    &\leq&\binom{n}{2}^{kg} \left(\frac{\tilde{\epsilon}}{\pi}\right)^{d^n} \leq \left(\frac{\tilde{\epsilon}}{2}\right)^{d^n} \label{eq:countingbound2},
\end{eqnarray}
where we have inserted $g$ in Eq.(\ref{eq:countingbound2}) and used the asymptotically sharp bound $\binom{n}{2}^{1/(2\ln n)}\leq e$.

(2) is proven along the same lines: we start with Eq.(\ref{eq:countingbound1}) with $\calD$ replaced by $\calD_\pm$ and where $\mu$ is now  the counting measure divided by $|\calD_\pm|=2^{d^n}$. Since two different unitaries $U,U'\in\calD_\pm$ satisfy $\|U-U'\|_\infty=2$, every ball of radius smaller than one contains at most one element of $\calD_\pm$. Consequently, $\mu(\calD_\pm\cap B_\epsilon)\leq 2^{-d^n}$. Inserting this and arguing like for Eq.(\ref{eq:countingbound2}) then leads to the claimed result.
\end{proof}
\note{add comment on efficient distinguishability}

\subsection{Maximally coherent states}\label{sec:maxcoh}
A state vector $|\varphi\rangle=\sum_j\varphi_j|j\rangle$ is called \emph{maximally coherent}
with respect to a given orthonormal product basis $\{|j\rangle\}_{j\in\bbZ_d^n}$ of $(\bbC^d)^{\otimes n}$  if all its amplitudes have the same modulus $|\varphi_j|=d^{-n/2}$. We denote the set of all maximally coherent state vectors in $(\bbC^d)^{\otimes n}$ by $\calS$. Since $\calS$ is the orbit of $ d^{-n/2}\sum_j |j\rangle$ under the diagonal unitary group $\calD$, the normalized  Haar measure of the latter induces a normalized measure on $\calS$. With a slight abuse of notation we will write $\mu$ for both.  
Note that despite the connection between $\calD$ and $\calS$, they behave very differently regarding the quantitative effect of parameter changes: whereas changing a single parameter suffices to move a diagonal unitary a constant $\epsilon$ (say $1/2$) away in operator norm, one has to change a large fraction of parameters, i.e. $\Omega(d^n)$, in order to move a maximally coherent state by $\epsilon=1/2$ in norm. 
Nevertheless a similar hardness-of-approximation result holds.
The following theorem shows that the subset of $\calS$ that can be approximated by circuits of sub-exponential size $O(d^n/\ln n)$ is double-exponentially small.  The same is true for the subset $\calS_\pm\subset\calS$ of $2^{d^n}$ maximally coherent states with real amplitudes $\varphi_j=\pm d^{-n/2}$.

\begin{theorem}[Hardness of approximating  maximally coherent states]\label{thm:hardmaxcoh}
    For $n,d\geq 2$ and a  universal gate set $\calG\subset U(d^2)$ with $k\in\bbN$ elements, let  $\calS_g$ be the set of maximally coherent states on $(\bbC^{d})^{\otimes n}$ that admit an $\epsilon$-approximation by a circuit of size $g$  in the sense that 
    \begin{equation*}
        \calS_g:=\big\{|\varphi\rangle\in\calS\;\big|\;\exists U\in U(d^n):\calC_0(U,\calG)\leq g, \big\| |\varphi\rangle-U|0\rangle^{\otimes n}\big\|\leq\epsilon\big\}.
    \end{equation*}
    If $g\leq d^n/(16k\ln n)$, then 
    \begin{eqnarray}
        \mu\big(\calS_g\big)&\leq&  (4e) \exp\left[-\frac{d^n}{16}\right],\qquad\forall \epsilon\in[0,1]\quad\text{and}\label{eq:harcohstatethm1}\\ 
        \frac{|\calS_g\cap\calS_\pm|}{|\calS_\pm|} &\leq & 2\exp\left[-\frac{d^n}{8}\right],\ \quad\qquad\forall\epsilon\in[0,3/4].\label{eq:harcohstatethm2}
    \end{eqnarray}
\end{theorem}
\begin{proof}
Define $D:=d^n$. Aiming at Eq.(\ref{eq:harcohstatethm1}), we begin with bounding how much of $\calS$ is covered by a single $\epsilon$-ball. To this end, consider an arbitrary unit vector  $\phi\in\bbC^D$ as the center of a ball $B_\epsilon$. Regarding $\psi$ as a $\mu$-uniform random variable with values in $\calS$, we can write
\begin{eqnarray}
    \mu\big(\calS\cap B_\epsilon\big) &=& \bbP_\psi\big[\|\phi-\psi\|\leq \epsilon\big]\nonumber\\
    &\leq& \bbP_\psi\Big[\big|\langle\phi,\psi\rangle\big|\geq 1-\frac{\epsilon^2}{2}\Big],
\end{eqnarray}
where we have used that $|\langle\phi,\psi\rangle|\geq \Re \langle\phi,\psi\rangle = 1- \frac12\|\phi-\psi\|^2$. That $\psi$ is $\mu$-uniformly random means that its $D$ components are, when multiplied by $\sqrt{D}$, i.i.d. random variables uniform on the complex unit circle---so-called \emph{Steinhaus variables}.
 Hence, we can apply the Hoeffding inequality for Steinhaus sums  from Cor.8.10 in \cite{foucart2013mathematical}, which yields
 \begin{eqnarray}
     \bbP_\psi\Big[\big|\langle\phi,\psi\rangle\big|\geq 1-\frac{\epsilon^2}{2}\Big] &\leq& D\Big(1-\frac{\epsilon^2}{2}\Big)^2\exp\left[1-D\Big(1-\frac{\epsilon^2}{2}\Big)^2\right]\nonumber\\
     &\leq& \frac{D}{4}\;e^{1-\frac{D}{4}}\;\leq\;4\exp\left[1-\frac{3D}{16}\right],\label{eq:Steinhausb2}
 \end{eqnarray}
 where the first inequality in Eq.(\ref{eq:Steinhausb2}) uses monotonicity in $\epsilon$ (for $D\geq 4$, $\epsilon\in[0,1]$) and sets $\epsilon=1$ and the last inequality exploits $1\leq e^{D/16} 16/D$.
 
 From here on, we can copy the corresponding part of the proof of Thm.\ref{thm:harddiagonals}, which after inserting the supposed upper bound on $g$ leads to
 \begin{eqnarray*}
     \mu(\calS_g) &\leq & \binom{n}{2}^{kg} \mu\big(\calS\cap B_\epsilon\big)\;\leq\;4 e^{D/8} e^{1-3D/16}\;=\; (4e)   \exp\left[-\frac{D}{16}\right].
 \end{eqnarray*}
 Finally, Eq.(\ref{eq:harcohstatethm2}) is proven along the same lines, with the only differences that $\mu$ is replaced by the normalized counting measure, the appropriate Hoeffding inequality is the one for Rademacher sums (Cor.8.8 in \cite{foucart2013mathematical}) and, for the sake of a simpler expression, we have used that $(1-\epsilon^2/2)^2>1/2$ holds for all $\epsilon\in[0,3/4]$.
\end{proof}

\subsection{Uniformly distributed sets of examples}
Finally, we come back to the explicit examples provided in Cor.\ref{cor:exexcc}. In order to obtain hardness-of-approximation results for them, we form families of those examples and show that these obey the same estimates as general diagonal unitaries and maximally coherent states (derived in  Sec.\ref{sec:approximatediagonal} and Sec.\ref{sec:maxcoh}, respectively). We begin with examples of unitaries.

Let us fix $n$ distinct primes $p_1,\ldots,p_n$ (e.g. the first $n$ primes) and consider the following set of unitaries on $(\bbC^d)^{\otimes n}$:
\begin{equation}\label{eq:Ut}
    U_t:=\sum_{j\in\bbZ_d^n}|j\rangle\langle j| \exp\left[it\prod_{k=1}^n p_k^{j_k/d}\right],\quad t\in\bbN.
\end{equation}
From Cor.\ref{cor:exexcc} we know that each of them satisfies $\calC_0(U_t)\geq d^{n-4}$. The following theorem shows that the vast majority of the $U_t$'s are also hard to approximate in the sense that for every $\epsilon\in[0,1)$ all but a double-exponentially small fraction of the $U_t$'s satisfy $\calC_\epsilon(U_t,\calG)=\Omega(d^n/\ln n)$. The reason for this is, loosely speaking, that the $U_t$'s are uniformly distributed within the set of all diagonal unitaries. As a consequence,  the same hardness-of-approximation result (i.e., Thm.\ref{thm:harddiagonals}) applies to asymptotically large sets of them. To phrase it more mathematically, the sequence of atomic probability measures $\mu_N(U_t):=1/N$ supported on $U_1,\ldots,U_N$ converges weakly to the Haar measure of the group of all diagonal unitaries in the limit   $N\rightarrow\infty$.
\begin{theorem}[Hardness of approximation -- unitary examples]\label{thm:hardunitaryex}
    For $\epsilon\in[0,1)$, a  universal gate set $\calG\subset U(d^2)$ with $k\in\bbN$ elements and 
    $$ g:=\left(\frac{\ln\frac{\pi}{2}}{2 k}\right)\frac{d^n}{\ln n},$$
    define by $r_N:=\frac1N|\big\{U_t\;|\;\calC_\epsilon(U_t,\calG)\leq g,\;t\leq N\big\}|$ the fraction of the first $N$ unitaries defined in Eq.(\ref{eq:Ut}) that have circuit complexity not larger than $g$. Then 
    \begin{equation}
        \lim_{N\rightarrow\infty} r_N \leq \big(\arcsin(\epsilon/2)\big)^{d^n}.
    \end{equation}
\end{theorem}
\begin{proof}
    The statement follows from Thm.\ref{thm:harddiagonals} as a consequence of Weyl's uniform version of  Kronecker's density theorem \cite{Weyl1916}. In order to see this, let us parametrize an arbitrary diagonal unitary $U={\rm diag}(e^{2\pi i x_1},e^{2\pi i x_2},\ldots)$ by its vector of phases $x$, which is an element of the unit cube $Q:=[0,1)^{d^n}$. The Lebesgue measure on $Q$ then corresponds to the Haar measure of the group of diagonal unitaries. It will be convenient to identify $Q$ with $\bbR^{d^n} \mod 1$, i.e., to allow for components outside the unit-interval and consider them modulo 1.
    
    Each $U_t$ is parametrized by a vector of the form $t\theta\in Q$ where $\theta_j:=\phi(j)/(2\pi)$ with $\phi(j)$ defined in Eq.(\ref{eq:mainexampleset}).
    As shown in Cor.\ref{cor:Besicovitch} in Appendix \ref{app:tnt}, the $d^n$ components of $\theta$ are linearly independent over $\bbQ$. Moreover,  they are all transcendental due to the division by $\pi$ and the fact that the $\phi(j)$'s are algebraic. Hence  the  numbers $1,\theta_1,\theta_2,\ldots$ are linearly independent over $\bbQ$ so that by Weyl's uniform distribution theorem (see p.48 in \cite{kuipers2012uniform}) the sequence $(t\theta)_{t\in\bbN}$ is uniformly distributed $\mod 1$ in $Q$. This means that for all Jordan measurable subsets $J\subseteq Q$ we have 
    \begin{equation}
        \lim_{N\rightarrow\infty}\frac1N \big|\big\{t\theta\;|\; t\theta\in J, t\leq N\big\}\big|=\lambda(J),
    \end{equation}
    where $\lambda$ is the Lebesgue measure and all components are understood $\mod 1$.
    Applying this to the specific set $J$ that corresponds to all diagonal unitaries with circuit complexity not larger than $g$, we obtain in the notation of Thm.\ref{thm:harddiagonals} and by using Eq.(\ref{eq:arcsinbound}):
    \begin{equation}
        \lim_{N\rightarrow\infty}\frac1N \big|\big\{U_t\;|\; U_t\in\calU_g, t\leq N\big\}\big|=\mu(\calD\cap \calU_g)\leq \big(\arcsin(\epsilon/2)\big)^{d^n}.
    \end{equation}
    The chosen $J$ is Jordan measurable since it is a finite union of at most $\binom{n}{2}^{kg}$ bounded convex, and thus Jordan measurable, subsets.
\end{proof}

Now let us turn to the state examples. Still assuming distinct primes $p_1,\ldots,p_n$ we define the following family of maximally coherent states:
\begin{equation}\label{eq:psit}
    \psi_t :=d^{-n/2}\sum_{j\in\bbZ_d^n}|j\rangle \exp\left[it\prod_{k=1}^n p_k^{j_k/d}\right],\quad t\in\bbN.
\end{equation}
Again, we know from Cor.\ref{cor:exexcc} that each of them satisfies $\calC_0(\psi_t)\geq d^{n-4}$. Apart from a double-exponentially small fraction, they are also hard to approximate:

\begin{theorem}[Hardness of approximation -- state examples]\label{thm:hardstateex}
    For $n,d\geq 2$, $\epsilon\in[0,1]$, a  universal gate set $\calG\subset U(d^2)$ with $k\in\bbN$ elements and 
    $$ g:=\left(\frac{1}{16 k}\right)\frac{d^n}{\ln n},$$
    define by $r_N:=\frac1N|\big\{\psi_t\;|\;\psi_t\in\calS_g,\;t\leq N\big\}|$ the fraction of the first $N$ examples defined in Eq.(\ref{eq:psit}) that admit an $\epsilon$-approximation by a circuit of size at most $g$. Then 
    \begin{equation}
        \lim_{N\rightarrow\infty} r_N \leq (4e) \exp\left[-\frac{d^n}{16}\right].
    \end{equation}
\end{theorem}
\begin{proof} ({\it sketch})
    Copying the proof of Thm.\ref{thm:hardunitaryex} we obtain that $r_N\rightarrow \mu(\calS_g)$ from which the result follows by inserting the upper bound on $\mu(\calS_g)$ from Thm.\ref{thm:hardmaxcoh}.
\end{proof}

Despite the intended greater physical relevance, the set-based approximation results of Thm.\ref{thm:hardunitaryex} and Thm.\ref{thm:hardstateex} lose some of their appeal compared to the individual exact-generation result of Cor.\ref{cor:exexcc}. First, there is the statistical nature of the results, which misses our goal to pinpoint explicit examples. Second, while Cor.\ref{cor:exexcc} clearly produces examples of exponential circuit complexity with constant description complexity, the hardness-of-approximation results require asymptotically large $t$ so that the description complexity of individual hard instances is not under control. 

It should also be noted that not all unitaries $U_t$ or states $\psi_t$ are hard to approximate---precisely due to the uniform density argument used in the proofs. For instance, for every $\epsilon>0$ there is a $t\in\bbN$ such that $\|U_t-\1\|\leq\epsilon$ and similarly $\|\psi_t-\psi_0\|\leq\epsilon$, where $\psi_0$ is a product state.

Remedying these shortcomings would be desirable, but might run up against complexity-theoretic obstacles.

\note{+discussion/conclusion (height,degree,depth,etc)}
\note{examples approximate 1 and +}

\vspace{0.3cm}

\noindent {\it Acknowledgments.} YJ thanks Aram Harrow for an insightful discussion, in particular for pointing to `algebraic independence'.

\vspace{0.3cm}

\noindent {\it Statements and Declarations.}
This work has been partially supported by the Deutsche Forschungsgemeinschaft (DFG, German Research Foundation) under Germany's Excellence Strategy EXC-2111 390814868 and via the SFB/Transregio 352. YJ acknowledges support from the TopMath Graduate Center of the TUM Graduate School and the TopMath Program of the Elite Network of Bavaria.

\newpage\appendix
\color{black}
\section{Results in transcendental number theory}\label{app:tnt}

In this section, we collect results in transcendental number theory that are used in the main text. For a general overview of the topic we refer to the textbooks of Baker \cite{baker_1975} or Murty and Rath \cite{MurtyRath}.

\begin{apptheorem}[Lindemann-Weierstrass, Thm. 1.4 in \cite{baker_1975}]\label{thm:LW}
If $\alpha_1,\dots,\alpha_n$ are algebraic numbers that are linearly independent over $\bbQ$, then $e^{\alpha_1}, \dots, e^{\alpha_n}$ are algebraically independent.
\end{apptheorem}

\begin{appcor}\label{cor:E}
If $S\subset\overline{\bbQ}$ is an $n$-dimensional vector space over $\bbQ$ and $E:=\{e^\lambda|\lambda\in S\}$, then  $\gamma(E)=n$.    
\end{appcor}
\begin{proof}
    The Lindemann-Weierstrass theorem \ref{thm:LW} implies that $\gamma(E)\geq n$. For the converse inequality, suppose that $\{a_1,\ldots,a_m\}\subset S$ fulfill a non-trivial linear relation $\sum_{k=1}^m c_k a_k=0$ for some $c_k\in\bbQ$. Then 
    $$\prod_{k=1}^m \left(e^{a_k}\right)^{c_k}=1, $$ which implies that the $\{e^{a_k}\}_{k=1}^m$ are algebraically dependent. Therefore, $\gamma(E)$ cannot exceed $n$.
\end{proof}
\begin{apptheorem}[Besicovitch, \cite{besicovitch1940linear}]\label{thm:Besicovitch}
	Let $p_1,p_2,\dots,p_s$ be distinct primes, $b_1,b_2,\dots,b_s$ positive integers not divisible by any of these primes and $a_i:=(b_i p_i)^{1/n_i}$ positive roots for $i=1,\ldots,s$ and $n_i\in\bbN$.
	 If $P\in\bbQ[x_1,x_2,\dots,x_s]$ is a polynomial with rational coefficients of degree less than or equal to $n_1-1$ with respect to $x_1$, less than or equal to $n_2-1$ with respect to $x_2$, and so on, then $P(a_1,a_2,\dots,a_s)=0$ can hold only if $P=0$. 
\end{apptheorem}

\begin{appcor}[see \cite{RICHARDS1974268} for a proof based on Galois theory]\label{cor:Besicovitch}
	Let $n,d\in\bbN$. For distinct prime numbers $p_1, \dots,p_n$, the following $d^n$ algebraic numbers are linearly independent over $\bbQ$:
	\begin{align}
	\phi(j):=\left(\prod\limits_{k=1}^{n}p_k^{j_k}\right)^{\frac{1}{d}},\quad j\in\bbZ_d^n.
	\end{align}
\end{appcor}
\begin{proof}
    This follows immediately from Besicovitch's theorem (\ref{thm:Besicovitch}) when setting $b_i=1$ and arguing by contradiction: suppose there would be a non-trivial linear relation of the form $\sum_{j\in\bbZ_d^n} c_j \phi(j)=0$ with $c_j\in\bbQ$, then a non-zero polynomial $P$ of the form that is excluded by Thm.\ref{thm:Besicovitch} would exist.
\end{proof}
For the following theorem, recall that the \emph{degree} of an algebraic number is the minimal degree of a monic polynomial $p\in\bbQ[x]$ that has the number as a root.

\begin{apptheorem}[Diaz \cite{Diaz}, Philippon \cite{Philippon}]\label{thm:DiazPhilippon}
If $\alpha\neq 0,1$ is algebraic and $\beta\in\overline{\bbQ}$ has degree  $d\geq 2$, then $S:=\{\alpha^\beta , \ldots, \alpha^{\beta^{d-1}}\}$ has $\gamma(S)\geq d/2$.
\end{apptheorem}
In fact, according to the Gel'fond-Schneider conjecture (see Chap.24 in \cite{MurtyRath}) it might be $\gamma(S)\geq d-1$.

\note{appendix on  stationary states?}

\newpage
\bibliographystyle{halpha}
{\footnotesize\bibliography{biblio}}
\vspace{0.3cm}
\end{document}